\documentclass[11pt]{article}

\usepackage{natbib}

 \usepackage{setspace} 
\DeclareMathAlphabet{\mathpzc}{OT1}{pzc}{m}{it}
\newcommand{\Prf}{\mathpzc{Pr}}

\newcommand{\LRf}{\mathpzc{Lr}}

\newcommand{\LR}{\textrm{LR}}

\usepackage[makeroom]{cancel}

\newenvironment{sciabstract}{%
\begin{quote} \bf}
{\end{quote}}

\usepackage[normalem]{ulem}	

\newcounter{lastnote}

\usepackage{amsthm,amsmath}
\RequirePackage[OT1]{fontenc}
\RequirePackage{amsthm,amsmath}

\RequirePackage[colorlinks,citecolor=blue,urlcolor=blue]{hyperref}
\usepackage[latin1] {inputenc}
\usepackage{vmargin}
\usepackage{multirow}

\usepackage{amsthm, dsfont, tipa}
\usepackage{graphics}

\newlength{\tempdim}

\usepackage{graphicx,float,multirow,subfigure}

\usepackage{tikz}
\usetikzlibrary{arrows,decorations.markings,chains,shapes.arrows,shapes.misc,fit, shapes}

 \tikzset{
  big arrow/.style={
    decoration={markings,mark=at position 1 with {\arrow[scale=2,#1]{>}}},
    postaction={decorate},
    shorten >=0.4pt},
  big arrow/.default=black}

\theoremstyle{plain}

\usepackage{amssymb}
\usepackage{epsfig}

\newtheorem{teo}{Theorem}

\title{Impact of model choice on LR assessment in case of rare haplotype match (frequentist approach)}

\author
{Giulia Cereda,$^{1\ast, 2}$ \\
\\
\normalsize{$^{1}$University of Lausanne, Faculty of Law, Criminal Justice and Public Administration}\\
\normalsize{$^{2}$Leiden University, Mathematical Institute}\\
\\
\normalsize{$^\ast$To whom correspondence should be addressed; E-mail: giulia.cereda7@gmail.com.}
}

\date{}

\begin{document} 


\baselineskip18pt


\maketitle


\begin{sciabstract}
The likelihood ratio (LR) measures the relative weight of forensic data regarding two hypotheses. Several levels of uncertainty arise if frequentist methods are chosen for its assessment: the assumed population model only approximates the true one and its parameters are estimated through a limited database. Moreover, it may be wise to discard part of data, especially that only indirectly related to the hypotheses. Different reductions define different LRs. Therefore, it is more sensible to talk about ``a'' LR instead of ``the'' LR, and the error involved in the estimation should be quantified. Two frequentist methods are proposed in the light of these points for the `rare type match problem', that is when a match between the perpetrator's and the suspect's DNA profile, never observed before in the database of reference, is to be evaluated.

\end{sciabstract}

\noindent
\emph{Key words}: Evidence evaluation, frequentist approach, likelihood ratio, rare type match, uncertainty, Y chromosome STR.


\section{Introduction}

One of the main challenges of forensic science is to evaluate how much some evidence can be helpful to discriminate between hypotheses of interest. For instance, a typical piece of evidence may be a DNA trace which is found at the crime scene and whose profile matches a known suspect's DNA profile. 
A couple of mutually exclusive hypotheses is typically defined, of the kind of `the crime stain came from the suspect' ($h_p$) and `the crime stain came from an unknown donor' ($h_d$). 
The largely accepted method to perform this evaluation is the calculation of the \emph{likelihood ratio}, a statistic that expresses the relative plausibility of the observations under the two hypotheses \citep{robertson:1995, evett:1998, aitken:2004,balding:2005, steele:2014}.

The definition of the likelihood ratio depends on whether a Bayesian or a frequentist approach is chosen.
In the Bayesian context, after a couple of hypotheses is given, the likelihood ratio is defined as
\begin{equation}
\label{eqa}
\LR=\frac{\Pr(D=d\mid H=h_p)}{\Pr(D=d \mid H=h_d)},
\end{equation}
where $\Pr$ is the Bayesian probability, reflecting the expert's belief on the joint distribution of the random variables of the model, namely $D$ (representing the data), $H$ (representing the hypotheses), and $\Theta$ (the nuisance parameter(s)).  

On the other hand, in a frequentist context, the nuisance parameter $\theta$ and the hypotheses $h$ are considered to be fixed (unknown) quantities. The frequentist probability (here denoted as $\mathcal{Pr}$) can be expressed in terms of the Bayesian $\Pr$, in the following way: $\mathcal{Pr}_\theta(\cdot \mid h):=\Pr(\cdot\mid \Theta=\theta, H=h),\ \forall h$. The frequentist likelihood ratio can be thus expressed as
\begin{equation}\label{e12}
\LRf_{\theta}=\frac{\mathcal{Pr}_{\theta}(D=d\mid h_p)}{\mathcal{Pr}_{\theta}(D=d \mid h_d)}.
\end{equation}

It is important to consider that different reductions of the data $D$ can be carried out, each corresponding to a different frequentist likelihood ratio. Moreover, unless we choose nonparametric solutions, a model choice is also performed, and  there are often parameters to be estimated. 
Hence, two further levels of uncertainty have to be added to the initial uncertainty regarding which hypothesis is the true one.

The main aim of this paper is to provide the message that, if a frequentist approach is chosen and an estimation is needed, (i) it is more sensible to talk about ``a'' likelihood ratio instead of ``the'' likelihood ratio, and (ii) a quantification of the error involved in the estimation of the likelihood ratio is to be provided along with the estimated value.

It is believed in the forensic field that the use of frequentist methods to assess the likelihood ratio is not coherent, since the likelihood ratio has to be used within the Bayes'  theorem context, as the way to update prior odds to posterior odds. However, frequentists may be interested as well in the likelihood ratio, seen as a tool to measure the evidential value of data, independently of the Bayes' theorem. Moreover, literature presents many approaches to calculate the likelihood ratio, wrongly defined as Bayesian, which in fact plug in Bayes estimates into a likelihood ratio defined in a frequentist way \citep[for a discussion, see][]{cereda:2015}. We thus believed that it is important to study and discuss the two approaches (the Bayesian and the frequentist) separately, in order to define coherent methodologies and avoid unnecessary hybrid methods. This is done in Section~\ref{fvb}.

In forensic science, a very challenging problem is the so-called \emph{rare type match}, the situation in which there is a match between the characteristics of some recovered material and the corresponding characteristics of the control material, but these characteristics have not been observed yet in previously collected samples (i.e., they do not occur in any existing database of interest for the case). This constitutes a problem because of the presence of a nuisance parameter that is (related to) the proportion of individuals (or items) in possess of the matching characteristic in a reference population: this proportion is, in standard frequentist practice, estimated using the relative frequency of the characteristic in a previously collected database. Thus, in case of rare type match, there's the need for different solutions.

This paper discusses two frequentist methods to provide a likelihood ratio in the rare type match case, based respectively on the parametric discrete Laplace method \citep{andersen:2013b}, and on a generalization of the nonparametric Good-Turing estimator \citep{good:1953}. The latter looks similar to Brenner's `$\kappa$-method' \citep{brenner:2010}, but is different inasmuch it does not need any assumption and provides two different frequencies, one for the prosecution's and one for the defense's point of view. We plan to compare the two methods in a future paper. 

More specifically, these two methods are here proposed as an answer to the problem of the rare Y-STR haplotype match: the situation in which the matching (and previously unseen) characteristic is a Y-STR profile.
Each of the two methods is analyzed in the light of points (i) and (ii) discussed above, by carefully specifying the data reduction, the chosen probability model, and with a discussion on the different levels of error involved in the estimations.

Sections~\ref{LR} and \ref{levels} draw out in depth the rationale behind points (i) and (ii) above, Section~\ref{ystr} describes the paradigmatic example of the rare Y-STR haplotype match problem, to which we will apply the discrete Laplace method (Section~\ref{DLm}), and the Generalized-Good method (Section~\ref{gg}) according to the guidelines exposed in the opening sections.

\section{Bayesian versus frequentist approach to likelihood ratio assessment} \label{fvb}

The task of a forensic statistician is to measure the extent to which some given data favors one hypothesis instead of the other.
For instance, the data at disposal may consist of a DNA trace found at the crime scene which matches a suspect's DNA profile, and of a database of collected DNA profiles from a reference population or past cases. This is a paradigmatic example to which, from now on, we will refer generically as ``the DNA example''. 
The prosecution and defense hypotheses are usually of the kind ``the trace has been left by the suspect'' ($h_p$) and ``the trace has been left by an unknown person'' ($h_d$).
Denote with $h\in \{h_d, h_p\}$ the unknown true hypothesis, and with $\theta$ the nuisance parameter involved in the assessment of the likelihood ratio. In the DNA example, the vector made of all the DNA frequencies can be thought of as the nuisance parameter $\theta$.
Notice that there is a difference between $h$ and $\theta$: one ($h$) is the parameter which we `test' through the likelihood ratio, the other ($\theta$) is a nuisance parameter involved in the likelihood ratio assessment. 
It is often possible to split the data $D$ into $E$, evidence directly related to the crime, and $B$, additional information not related to the crime and only pertaining to the nuisance parameter $\theta$. In the DNA example, we can take as $E$ the couple of matching profiles (that of the trace and that of the suspect) and as $B$ the database of reference.
$D$, $E$, and $B$ can be regarded as random variables, such that $D=(E,B)$.

Bayesian and frequentist methods differ in how they consider the parameters $\theta$ and $h$. In a Bayesian context they are modelled through random variables $\Theta$ and $H$, which are given prior distributions $p(\theta)$ and $p(h)$. Frequentists consider them as fixed (i.e., without distribution) unknown quantities. 
Regardless of the type of approach which is chosen, some model assumptions concerning $E$ and $B$, $\theta$ and $h$ can be made:

\begin{description}
\item{\textbf{a.}} The distribution of $B$ given $h$ and $\theta$, only depends on $\theta$.
\item{\textbf{b.}} $B$ is independent of $E$, given $h$ and $\theta$.
\end{description}

In the DNA example, condition \textbf{a} holds if for instance the database is collected before the crime, since the sampling mechanism to obtain the database of reference is independent of which hypothesis is correct. Condition \textbf{b} holds if the suspect has been found on the ground of different evidence that has nothing to do with DNA.

\subsection{The Bayesian approach}\label{bay}
\begin{figure}[htbp]
\begin{center}
\begin{tikzpicture}
  \node[draw, ellipse, minimum width=1.2cm]                (a) at (0,0)  { $\Theta$ };
  \node [draw, ellipse, minimum width=1.2cm]                  (b) at (2,0)  { $H$ };
  \node [draw, ellipse, minimum width=1.2cm]                      (c) at (2,-2) { $E$};
  \node [draw, ellipse, minimum width=1.2cm]              (d) at (0,-2) { $B$};
   \draw[black, big arrow] (a) -- (c);
 \draw[black, big arrow] (b) -- (c);
 \draw[black, big arrow] (a) -- (d);
\end{tikzpicture}
\caption{Bayesian network representing the dependency relations between $E$ (evidence of the case), $B$ (background data), $\Theta$ (nuisance parameter), and $H$ (hypotheses of interest).}\label{bnet_t1}

\end{center}
\end{figure}
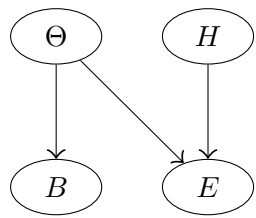

A full Bayesian model is defined by giving the prior joint probability distribution for all the random variables of the model (here $E$, $B$, $H$ and $\Theta$). 
It can be represented by the Bayesian network of Figure~\ref{bnet_t1}, which is in turn equivalent to the following Bayesian reformulation of conditions \textbf{a}, and \textbf{b}, with a third additional condition:
\begin{description}
\item{\textbf{Bayesian a.}} $B$ is conditionally independent of $H$ given $\Theta$.
\item{\textbf{Bayesian b.}} $B$ is conditionally independent of $E$ given $\Theta$ and $H$.
\item{\textbf{Bayesian c.}} $\Theta$ is unconditionally independent of $H$.
\end{description}
Condition \textbf{Bayesian c}\ is guaranteed for instance if prior beliefs on $\theta$ and on $h$ are assessed by people with different responsibilities and tasks: a judge for $h$ and a forensic DNA expert (or a statistician) for $\theta$. 
The joint prior can be factorized as follows, by looking at the structure of the Bayesian network or, equivalently, using the three conditions above:
$p(\theta, h, b, e) = p(\theta)p(h) p(b|\theta) p(e|\theta, h).$
By choosing a prior distribution for $\theta$ and $h$ which reflects expert's beliefs, the Bayesian probability is an expression of the subjective credence of the experts. The distribution of all other variables given $\theta$ and $h$ is defined by the structure of the model, and needs no subjective assessment.

The Bayesian likelihood ratio can be derived in the following way:
\begin{align*}
\LR=&\frac{\Pr(E=e, B=b\mid H=h_p)}{\Pr(E=e, B=e\mid H=h_d)}=\frac{\Pr(E=e\mid B=b, H=h_p)}{\Pr(E=e\mid B=b, H=h_d)}
= \frac{\int p(e \mid b, h_p, \theta)\, p(\theta \mid  b, h_p) \text{d}\theta}{\int p(e \mid  b, h_d, \theta)p(\theta \mid  b, h_d) \text{d}\theta }
\\=&  \frac{\int \theta\, p(\theta \mid  b)\, \text{d}\theta}{\int \theta^2\, p(\theta \mid b)\, \text{d}\theta}
= \frac{\mathbb{E}(\Theta \mid B=b)}{\mathbb{E}(\Theta^2\mid B=b)}.
\end{align*}
Some simplifications have been carried out because of conditions \textbf{a}, \textbf{b}, and \textbf{c}. 
Since it is possible to marginalize out over all values of $\Theta$, using its distribution, there's no need to estimate the likelihood ratio, or to account for uncertainties, if a proper full Bayesian approach is chosen.

In the rest of the paper we only focus on frequentist methods to solve the rare Y-STR haplotype match problem, but a companion paper presents a similar study on Bayesian methods \citep{cereda:2015}. 

\subsection{The frequentist perspective}\label{freq}
The difference between frequentist and Bayesian methods regards parameters $h$ and $\theta$: for a frequentist they are fixed quantities, whose values correspond to, respectively, the unknown true value of $\theta$ and the correct hypothesis. One can see frequentist models as Bayesian models where the distributions chosen for $\Theta$ and $H$ give probability one to values $\theta$ and $h$, respectively. Also, one can express the frequentist probability $\Prf$ in terms of the Bayesian probability $\Pr$, in the following way: 
$\Prf(\cdot\mid h):= \Prf_{\theta}(\cdot\mid h)= \Pr(\cdot | H=h, \Theta=\theta).$
For frequentist statisticians, there is a true, `physical' probability which governs the situation at hand: according to the prosecution this true probability is $\Prf_{\theta}(\cdot\mid h_p)$, while according to the defense it is $\Prf_{\theta}(\cdot\mid h_d)$,  with $\theta$ set to its true (unknown) value. 

Conditions \textbf{a} and \textbf{b} can be rephrased, in a frequentist language as:

\begin{description}
\item{\textbf{Frequentist a.}} $\Prf_{\theta}(B=b \mid  h_p)=\Prf_{\theta}(B=b \mid  h_d)$, for all $\theta$ and $b$.
\item{\textbf{Frequentist b.}} $\Prf_{\theta}(E=e \mid B=b, h)=\Prf_{\theta}(E=e \mid  h)$, for all $\theta, h, e$, and $b$.
\end{description}

It holds that:
\begin{equation*}\label{eq22x}
\begin{aligned}
	\LRf	&=\frac{\Prf(D=d \mid h_p)}{\Prf(D=d \mid h_d)}=\frac{\Prf(E=e,B=b \mid h_p)}{\Prf(E=e,B=b \mid h_d)}
	=\frac{\Prf(E=e \mid B=b,h_p)}{\Prf(E=e \mid B=b,h_d)}\frac{\Prf(B=b \mid h_p)}{\Prf(B=b \mid h_d)}.
\end{aligned}
\end{equation*}

The index $\theta$ has been omitted for ease of notation. Thanks to conditions \textbf{Frequentist a} and \textbf{b}, the likelihood ratio can be expressed as
\begin{equation}
\label{eq22}
\LRf=\frac{\Prf(E=e \mid h_p)}{\Prf(E=e \mid h_d)}.
\end{equation}

Even though the two alternative ways of writing the likelihood ratio expressed by equations~\eqref{e12} and~\eqref{eq22} are theoretically different, and mean two different things, they have the same value. This implies that part of the information, namely $B$, is not useful to discriminate between the two hypotheses of interest.
Stated otherwise, when knowing $\theta$, $B$ is irrelevant to determine the likelihood ratio, i.e.\ to decide about parameter $h$. However, it may play an important role in the estimation of parameter $\theta$.
For instance, getting back to the DNA example, the database ($B$) is often useful to estimate the frequencies of the different haplotypes.

Notice that, in order for \eqref{eq22} to hold, \textbf{b} can be modified to something less strong: 
\begin{description}

\item{\textbf{Frequentist b$^*$.}} $\displaystyle\frac{\Prf_{\theta}(E=e \mid B=b, h_p)}{\Prf_{\theta}(E=e \mid B=b, h_d)}=\frac{\Prf_{\theta}(E=e \mid  h_p)}{\Prf_{\theta}(E=e \mid  h_d)}$ for all $e$, $b$, and $\theta$.
\end{description}
which is equivalent to ask that updating the likelihood ratio for the observation of $B$ to take into account the observation of $E$, does not change anything. 

Furthermore, while conditions \textbf{a} and \textbf{b$^*$} imply \eqref{eq22}, the converse is not true. Formulation \eqref{eq22} is instead equivalent to a weaker condition, that is:
\begin{description}
\item[Frequentist c.]$ \Prf_{\theta}(B=b \mid E=e, h_p)=\Prf_{\theta}(B=b \mid E=e, h_d)$, for all $\theta$. \end{description}

This can be seen by the following alternative development of the likelihood ratio ($\theta$ omitted):
\begin{equation}
\begin{aligned} 
\label{eq22xnj}
\LRf&=\frac{\Prf(D=d \mid h_p)}{\Prf(D=d \mid h_d)}=\frac{\Prf(B=b \mid E=e,h_p)}{\Prf(B=b \mid E=e,h_d)}\frac{\Prf(E=e \mid h_p)}{\Prf(E=e \mid h_d)}=
\frac{\Prf(E=e \mid h_p)}{\Prf(E=e \mid h_d)}.
\end{aligned}
\end{equation}
It follows that:
\begin{equation}
\label{eq22xmm}
 \textbf{c}\Leftrightarrow \LR=\frac{\Prf(E=e \mid h_p)}{\Prf(E=e \mid h_d)}.  \end{equation}

Notice that frequentists use a likelihood ratio $\LRf_{\theta}$, which can be written in terms of the Bayesian LR as $\LR|\Theta=\theta$ (read ``LR given $\theta$''), and attempt to get close to $\theta$ by choosing some estimator $\widehat{\theta}$. This leads to the so-called \emph{plug-in estimator} $\widehat{\LRf_{\theta}}=\LRf_{\widehat{\theta}}=\LR|(\Theta=\widehat{\theta})$. However, that's not the only option, as we will see for the method explained in Section~\ref{gg}.

It is important to notice that the frequentist approach may be represented by the same Bayesian network of Figure~\ref{bnet_t1}, where the states of nodes $\Theta$ and $H$ are instantiated to particular values $\theta$ and $h$, respectively. This shows that actually the two approaches don't disagree on the structure of the model regarding $E$ and $B$. Only, Bayesians add ingredients to the model by allowing $\Theta$ and $H$ to have a distribution. Stated otherwise, the Bayesian approach is defined by the very same frequentist conditions  \textbf{a} and \textbf{b}, with the addition of condition \textbf{c} about the independence of $\Theta$ and $H$.

\section{ Data reduction}\label{LR}
 Let us denote with $\mathcal{D}$ all the data given to the expert in the form of a dossier, which he has to ``translate'' into a well-defined mathematical object. 
To evaluate the entirety of the data at the expert's disposal is often a delusion, from which the need for a reduction of $\mathcal{D}$ to something less informative, but of more feasible evaluation, which we denote as $D$. 
Often the database contains only information about a limited number of loci, and this implies that information about other loci of the crime stain can't be used. This constitutes already a first reduction of the data.
Other kinds of reductions are performed in order to gain in terms of precision of the estimates. Especially in a situation with many nuisance parameters, it can be wise to discard the part of data which primarily tells us about the nuisance parameters, and only indirectly about the ultimate question of interest (i.e., which hypothesis is more likely to be true).
In fact, it could be very wise to reduce the data $\mathcal{D}$ to a much smaller amount of information, because the likelihood ratio based on the data reduction is much more precisely estimated than one based on all data. However, there's a limit to this: the reduction of $\mathcal{D}$ into $D$ comes with a cost: the stronger the reduction, the less the corresponding likelihood ratio value is discriminating of the two hypotheses, because less information is less powerful to that purpose. We have to make a compromise between a gain in terms of precision and a loss in strength of the evidence. This will be discussed more in detail in Section~\ref{trade}.

Once a particular reduction $D$ has been defined, the frequentist likelihood ratio ($\LRf$) can be defined as in \eqref{e12}.
It is easy to understand that there isn't a unique way to reduce $\mathcal{D}$ and that each choice entails the definition of a different likelihood ratio. For instance, in the DNA example, one can consider a profile made of more or fewer loci. Another kind of reduction will be presented in Section~\ref{gg}.
Different choices of $D\subsetneq \mathcal{D}$ lead to different likelihood ratios. Therefore, \emph{it is better to refer to ``a'' likelihood ratio instead of to ``the'' likelihood ratio.} This was already stated in \citet{dawid:2001}, even though regarding hypotheses instead of data.

In the literature different choices of $D\subsetneq \mathcal{D}$ and `Pr' are proposed, each corresponding to a different likelihood ratio to be estimated. These choices are often only implicit and one of the aims of this research is to make explicit the reduction which corresponds to two selected methods, by looking for the corresponding $E$ and $B$.

\section{Different levels of uncertainty}\label{levels}

The likelihood ratio measures the relative strength of support given by the data to a hypothesis over an alternative. Clearly, it is useful when there is uncertainty about which of the two hypotheses is true (to be more precise, it may also be the case that none of the alternatives is correct, and the likelihood ratio continues to be meaningful). 
Along with this first basic initial uncertainty about the state of the affairs, two more levels of uncertainty arise in the attempt of calculating the likelihood ratio.

For a frequentist statistician, the likelihood ratio is a ratio of probabilities based usually on a model $\mathcal{M}$ which is at best only a good approximation to the truth. 
Moreover, they have to estimate parameters of that model by fitting it to the data in some database.
Stated otherwise, after a particular choice of what is the data $D$ to be considered, a population model is to be chosen and its parameters estimated using a limited sample. Some forensic literature \citep{morrison:2010, stoel:2012, curran:2002, curran:2005} already pointed out the necessity for uncertainty assessment in the likelihood ratio estimation, even though they don't differentiate among levels.
On the other hand, for a true Bayesian statistician there's no need for estimation, and no additional levels of uncertainty to be added, since the definition of the Bayesian $\Pr$ already includes not only beliefs about chances when picking people from that population, but also beliefs about parameters of the models, and beliefs about models.

This discussion may hopefully put an end to the debate as to whether it makes sense to talk about  `estimation' and `uncertainty assessment' for the likelihood ratio.
\citet{stoel:2012}  believe that ``there are strong arguments
for the notion of a ``true'' but unknown value of the likelihood ratio, given the relevant hypotheses and
background information, and that it is important to consider the uncertainty. Ignoring the
uncertainty can be strongly misleading''. This point of view is also shared in \citet{sjerps:2015}. On the other hand, to talk about estimation of the likelihood ratio is defined as ``internally inconsistent, and hence misconceived'' by \citet{taroni:2015}.
Both the points of view are correct, if correctly put into context: if a frequentist approach is chosen it is sensible to talk about `estimation' and to deal with uncertainty assessment. On the other hand, in a full Bayesian context, they are misplaced.

Notice that Bayesianism is theoretically a very powerful interpretation of probability, but when it comes to applying Bayesian theory for practical purposes, even the most fervent Bayesian has to strike a balance between what is feasible and what is theoretically right and coherent according to the Bayesian perspective. He typically chooses a particular model as the correct one (as frequentists do), and/or he has to put convenient (rather than realistic) prior distributions on the parameters. 
Hence, whether Bayesian or frequentist approaches are chosen, the attempt to produce the likelihood ratio leads to several levels of uncertainty which should be accounted for.

We will now discuss the two additional levels of uncertainty mentioned before. The second level of uncertainty pertains to the choice of a particular population model, which is only an approximation of the truth. This level of uncertainty may be reduced using nonparametric methods, that are based on fewer assumptions.

Given a particular population model, the third level of uncertainty pertains to the fact that the population parameters are not known. This may involve estimation of parameters (such as in the discrete Laplace method of Section~\ref{DLm}) or the direct estimation of the probabilities of interest (as in the Generalized-Good method described in Section~\ref{gg}) and the quality of the estimates severely depends on the size of the available databases. This level of uncertainty pertains both to parametric and nonparametric methods. 

The evidential value reported depends on all the levels of uncertainty which afflict the estimation of the likelihood ratio. Thus, it is of the utmost importance to report the likelihood ratio value along with (1) an explicit definition of which data $D$ we want to evaluate through that likelihood ratio, and (2) a discussion (and if possible a quantification) of the levels of uncertainty that afflict the reported value.

\subsection{Estimating the weight of evidence}
Instead of estimating the likelihood ratio, it is more sensible to directly estimate its logarithm, sometimes called \emph{relevance ratio} or \emph{weight of evidence} \citep{good:1950, aitken:1998, aitken:2004}.
This is because the interpretation of the likelihood ratio values goes through orders of magnitude 10, and when a value is reported, it is important to control the relative error, rather than the absolute error. In fact, the first is meaningful in itself while the second depends on the particular value of the likelihood ratio. 
For the very same reasons why the verbal equivalent scale \citep{aitken:1998} is based on logarithm.
Furthermore, both the odds form of Bayes' theorem and the formula to combine likelihood ratios from independent pieces of evidence involve a multiplicative relationship that becomes a handier additive relation if logarithm is taken \citep{schum:1994}. 
Moreover, the logarithm helps in presenting large numbers in a compact way, of more easy comprehension, and it is symmetric with respect to prosecution's and defense's hypothesis: this may be useful if one wants to invert the weight of evidence to consider the defense's proposition \citep{aitken:2004}.

\section{The rare Y-STR haplotype problem}\label{ystr}

Consider the situation in which a piece of evidence is recovered at the crime scene, and a suspect turns out to have the same analyzed characteristics (for instance the same DNA profile) as the crime scene evidence. 
The prosecution claims that the suspect left the evidence, defense claims that someone else (with the same DNA profile) left it. 
The capability of the match to discriminate between the competing hypotheses is evaluated by comparing how probable it is under each of the hypotheses. This depends on the proportion of individuals in possess of the same profile in the population of possible perpetrators: the rarer the profile the more the suspect is in trouble. 
This proportion is usually unknown, the only available data being a sample of DNA profiles from the population, in the form of a reference database. 
The \emph{naive estimator} uses the relative frequency of the profile in the database as an estimate for $\theta$. Problems arise when this frequency is 0, the so-called ``rare type match''. This problem is so substantial that it has been defined ``the fundamental problem of forensic mathematics'' by \citet{brenner:2010}. As an alternative to the empirical frequency estimator, one can use the \emph{add-constant} estimators, which adds a constant to the count of each type, included the unseen ones. The most well known is the \emph{add-one} estimator, due to \citet{laplace:1814}, and the \emph{add-half} estimator of \citet{krichevsky:1981}. However, to use these methods one needs to know the number of possible unseen types and there are problems if this number is large compared to the sample size (see \citet{gale:1994} for additional discussion). 
Another possibility is the `rule of three', proposed by \citet{louis:1981}. It states that $3/n$ is a good approximation of the $95\%$ upper bound for the frequency, if $n$ is the size of the database. 

Of interest for this paper is the nonparametric \emph{Good-Turing estimator} of \citet{good:1953}, based on an intuition on A.\ M.\ Turing. It is an estimator for the total unobserved probability mass which is based on the proportion of singletons in the sample. For a comparison between \emph{add one} and \emph{Good-Turing} estimator, see \citet{orlitsky:2003}.

The \emph{naive estimator} and the \emph{Good-Turing estimator} are
in some sense complementary \citep{anevski:2013}: the first gives a good estimate for the observed types and the second for the probability mass of the unobserved ones. 
Lastly, the \emph{high profile estimator}, introduced by \citet{orlitsky:2004}, extends the tail of the \emph{naive estimator} to the region of unobserved types. This estimator has been improved by \citet{anevski:2013} that also provides the consistency proof. 

The rare type match problem is common, for instance, in case a new kind of forensic evidence is involved, and for which the available database size is still limited. One example is the case of DIP-STR markers \citep[e.g.][]{cereda:2014b}.
The same happens when Y-chromosome (or mitochondrial) DNA profiles are used: because of the lack of recombination involved when offspring DNA is generated from the DNA of the parents, each haplotype must be treated as a unit (the
match probability can't be obtained by multiplication across loci) and the set of possible haplotypes is extremely
large. As a consequence, most of the Y-STR haplotypes are not represented in the database.

In the rest of the paper, Y-STR data will be retained as an extreme but common and important way in
which the problem of assessing the evidential value of rare type match can arise. Literature provides some examples of approaches to evaluate it for the rare Y-STR haplotypes match: \citet{egeland:2008}, the $\kappa$ method \citet{brenner:2010, brenner:2014},  the coalescent theory method \citep{andersen:2013}, the haplotype surveying method \citep{roewer:2000,krawczak:2001,willuweit:2011}, and the discrete Laplace method \citep{andersen:2013b} (not directly proposed for the rare haplotype case but usable for that purpose).
As already mentioned, \citet{cereda:2015} discusses the full Bayesian approach to this problem. 

Bayesian nonparametric estimators for the probability of observing a new type have been proposed by \citet[e.g.][]{tiwari:1989, lijoi:2007, favaro:2009}.
However, for the likelihood ratio assessment it is required not only the probability of observing a new species but also the probability of observing this same species twice (according to the defense the crime stain profile and the suspect profile are two independent observations). \citet{cereda:2015c} is the first paper that addresses the problem of likelihood ratio assessment in the rare haplotype case using Bayesian nonparametric models.

The present paper analyses two frequentist methods, the discrete Laplace method, and a generalization of the Good-Turing, making explicit the corresponding definitions of $D$, $E$, and $B$, and providing a study on the different levels of uncertainty arising for each.

\section{The discrete Laplace Method}\label{DLm}

A discrete random variable $X$ is said to follow the discrete Laplace distribution $\text{\text{DL}}(p,y)$, with dispersion parameter $p\in (0,1)$, and location parameter $y\in \mathbb{Z}$, if its probability density is defined as
$$f(x \mid p,y)=\left(\frac{1-p}{1+p}\right)p^{|x-y|}, \quad \forall x\in \mathds{Z}.$$

This  is used in \citet{andersen:2013b} to model the distribution of single locus Y-STR haplotype in some subpopulation, which is thus assumed to be distributed around a modal allele (represented by the location parameter $y$).

Each haplotype is actually composed by $r$ loci. Let denote with $\mathbf{X}=(X_1,X_2,..., X_r)$ the random variable which describes an $r$-loci haplotype configuration. Moreover, there may be $c$ different subpopulations to take into consideration. 
By making the strong assumption of independence between loci, within the same subpopulation, the following density is used to describe the probability that $\mathbf{X}=\mathbf{x}$: 

\[f(\mathbf{x}\mid \{\mathbf{y_j}\}_j,\{\mathbf{p_j}\}_j )=\sum_{j=1}^{c}\tau_j \prod_{k=1}^{r}f(x_k \mid y_{jk}, p_{jk}),\]
where, for each $j$, $\tau_j$ is the probability a priori of generating from the $j$th subpopulation, while $\mathbf{p_j}=(p_{j1},p_{j2},..., p_{jr})$ and $\mathbf{y_j}=(y_{j1},y_{j2},..., y_{jr})$ represent the dispersion and location parameters, respectively, of the $j$th subpopulation. 
\citet{andersen:2013b} propose to estimate all these parameters by using using the EM algorithm \citep{dempster:1977}. The initial subpopulation centres are chosen by PAM algorithm \citep{kaufman:2009} and the number of them by the Bayesian Information Criteria (BIC) \citep{schwarz:1978}.

\subsection{The choice of $D$ in the discrete Laplace Method} 

The choice of $D$ which underlies the discrete Laplace method, when used to address the rare haplotype match problem is:

\begin{itemize}
\item$D_{\text{DL}}$ = the particular haplotype $x$ of the suspect and of the stain, along with a database which is a sample from the population of possible perpetrators.
\end{itemize}

This method allows one to evaluate the data in the light of the usual hypotheses of interest in the DNA example (see Section~\ref{fvb}).
$D_{\text{DL}}$ can be split into $E_{\text{DL}}$ and $B_{\text{DL}}$, in the following way:
\begin{itemize}
\item $E_{\text{DL}}$ = the particular haplotype $x$ of the stain ($E_t$) and of the suspect ($E_s$). 
\item $B_{\text{DL}}$ = random sample from the population of possible perpetrators (i.e.\ database).
\end{itemize}
The vector containing the frequencies of all haplotypes in the population of reference can be thought of as the nuisance parameter $\theta$ of this model.
Conditions \textbf{a.}\ and \textbf{b.}\ presented in Section~\ref{fvb} are valid for $E_{\text{DL}}$, $B_{\text{DL}}$, $\theta$, and $h$, thus the following likelihood ratio (where $\theta$ is again omitted) corresponds to this choice of data, evidence, background, and model:
\begin{equation}
\begin{aligned}
\LRf_{\text{DL}}&=\frac{\Prf(D_{\text{DL}}=d \mid h_p)}{\Prf(D_{\text{DL}}=d \mid h_d)}=\frac{\Prf(E_t=x\mid E_s=x,h_p)\Prf(E_s=x \mid h_p)}{\Prf(E_t=x \mid E_s=x,h_d)\Prf(E_s=x \mid h_d)}\\&=
\frac{\Prf(E_t=x \mid E_s=x,h_p)}{\Prf(E_t=x \mid h_d)}=\frac{1}{f_x}.
\end{aligned}
\end{equation}
Here, $f_x$ is the frequency of the haplotype $x$ in the population of reference. The second equality is due to conditions \textbf{a} and \textbf{b} discussed in Section~\ref{freq}, while the fourth one is justified by the fact that the distribution of the haplotype of the suspect does not depend on which hypothesis is correct, and that, when $\theta$ is fixed (as in the frequentist approach which we are considering) and under $h_d$, $E_t$ is independent of $E_s$. 
The weight of evidence is thus
\begin{equation}\label{eqs}
\log_{10}{\LRf_{\text{DL}}}= \log_{10}\frac{1}{f_x}. 
\end{equation}
The frequency $f_x$ can be estimated by $\widehat{f_x}$, using the discrete Laplace method. This brings to the following plug-in estimator for $\log_{10}\LRf_{\text{DL}}$:
$$\widehat{\log_{10}{\LRf_{\text{DL}}}}= \log_{10}\frac{1}{\widehat{f_x}}.$$

Notice that the discrete Laplace method uses the database to estimate the number of subpopulations and all the parameters in the model, and this is where $B_{\text{DL}}$ comes into play again.

\subsection{Quantifying the uncertainty of the discrete Laplace method}\label{q-dis}

We quantify the uncertainty of this method comparing the distribution of $\widehat{\log_{10}{\LRf_{\text{DL}}}}=\log_{10}{\displaystyle \frac{1}{\widehat{f_x}}}$ with the distribution of the ``true'' 
$\log_{10}{\LRf_{\text{DL}}}= \log_{10}{\displaystyle \frac{1}{f_x}}.$
$f_x$ is not known, but we have a database of approximately 19,000 Y-STR 23-loci profiles from 129 different locations in 51 countries in Europe \citep{purps:2014}\footnote{A clean version of the database is provided by Mikkel Meyer Andersen (\url{http://people.math.aau.dk/~mikl/?p=y23}).}, which we can pretend contains the whole population of interest for our case. 
We will consider only 7 loci out of 23 and perform the following experiment: we sample a small database of size $N=100$, along with a new haplotype (not observed in the small database), and calculate the estimate $\log_{10}\frac{1}{\widehat{f_x}}$.
Then, we can use the relative frequency of the haplotype $x$ in the big database as the true one, $f_x$ to obtain $\log_{10}\frac{1}{f_x}$.

This process can be repeated many times (for instance $M=1000$  samplings of small databases of size $N=100$ and, for each, a never observed haplotype). 

In estimating $\log_{10}{\LRf_{\text{DL}}}$ via $\widehat{f_x}$, one has the choice between adding the haplotype $x$ to the small database before estimating 
parameters of the discrete Laplace distribution, or not. In a full Bayesian approach the right thing to do is to add the profile to the database. This is shown in \citet{cereda:2015}, and we believe that it is the good thing to do also in a frequentist framework. In fact, experiments show that to add or not the haplotype to the database does not make much difference.

Table~\ref{tab1f} and Figure~\ref{figbru} (left part) compare the distributions of $\log_{10}{\LRf}$ and $\widehat{\log_{10}{\LRf_{\text{DL}}}}$, using 7 loci. The same experiment has been carried out for 10 and 3 loci, but not reported in details. 

\begin{table}[htbp]
\begin{center}
    \begin{tabular}{|l|c|c|c|c|c|c|c|}
    \hline
 &  Min   & 1st Qu. & Median & Mean  & 3rd Qu. & Max   & s.d.   \\
  \hline
	  $\log_{10}\LRf_{\text{DL}}$   			&   1.305	&	2.733 	&   	3.277 	&  	3.272 	& 3.800 	&   4.277	&	0.666   \\
					 $\widehat{\log_{10}\LRf_{\text{DL}}}$    	& 1.432   	&	3.441  	& 	4.061 	&  	4.114 	&  4.750 	&  8.452  	& 	1.017  \\
					Error $e_{\text{DL}}$					& -1.37 	&  	0.217 	& 	0.807	& 	0.842	&  1.39	&  4.476 	& 	0.863 \\				

    \hline
     \end{tabular}

    \caption{Summaries of the distribution of  $\log_{10}{\LRf_{\text{DL}}}$, $\widehat{\log_{10}{\LRf_{\text{DL}}}}$, and of the error $e_{\text{DL}}$. }\label{tab1f}
\end{center}
\end{table}

\begin{figure}[htbp]
\medskip

\hspace{0.6\baselineskip}\hfil

\makebox[0.45\textwidth]{ $\widehat{\log_{10}\LRf_{\text{DL}}}$ and $\log_{10}\LRf_{\text{DL}}$}
\makebox[0.45\textwidth]{ Error }

\settoheight{\tempdim}{\includegraphics[width=0.4\textwidth]{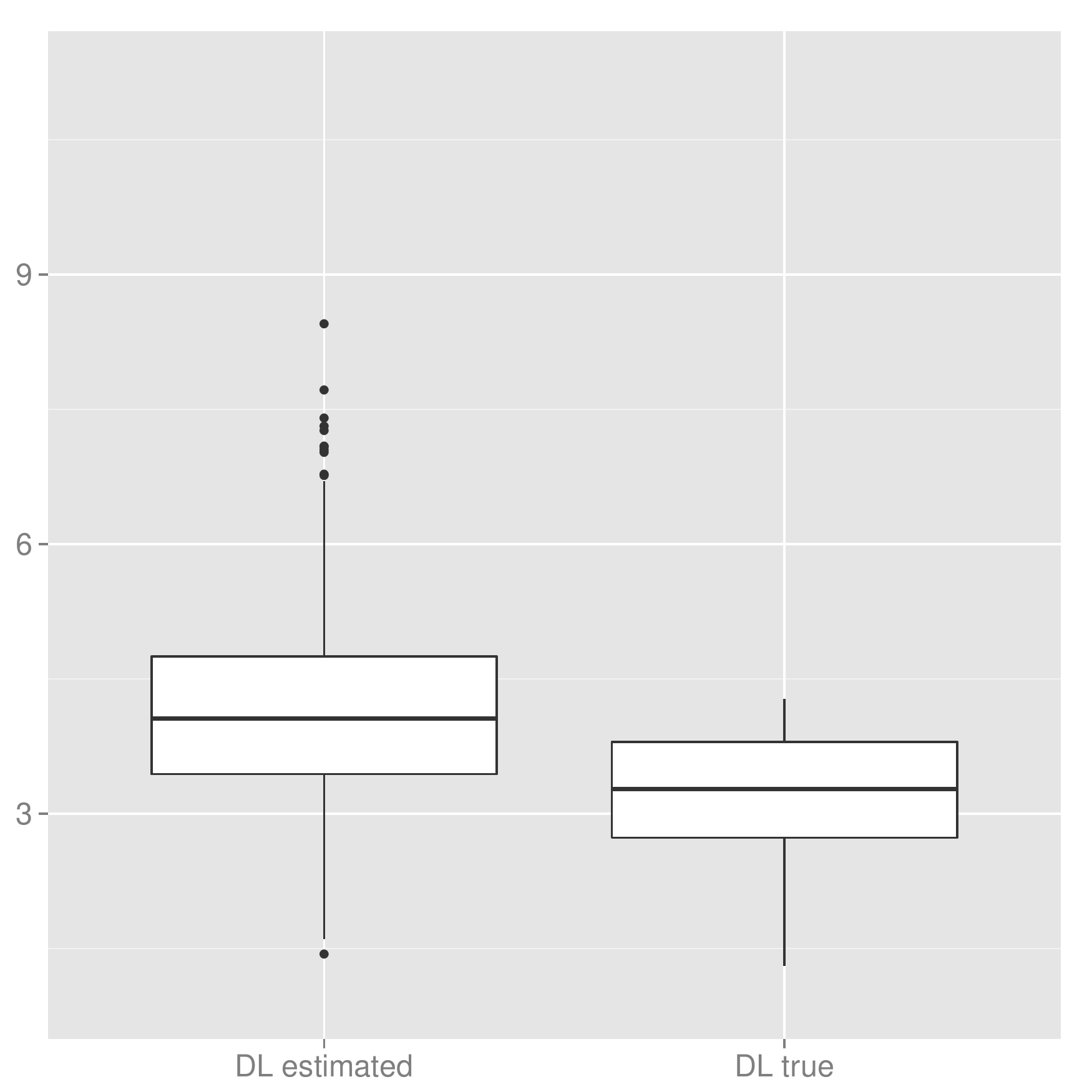}}
\subfigure[]{\includegraphics[trim=0cm 0cm 0cm 1.5cm, clip=true, width=0.45\textwidth]{Figure_3a.pdf}}
\subfigure[]{\includegraphics[trim=0cm 0cm 0cm 1.5cm, clip=true, width=0.45\textwidth]{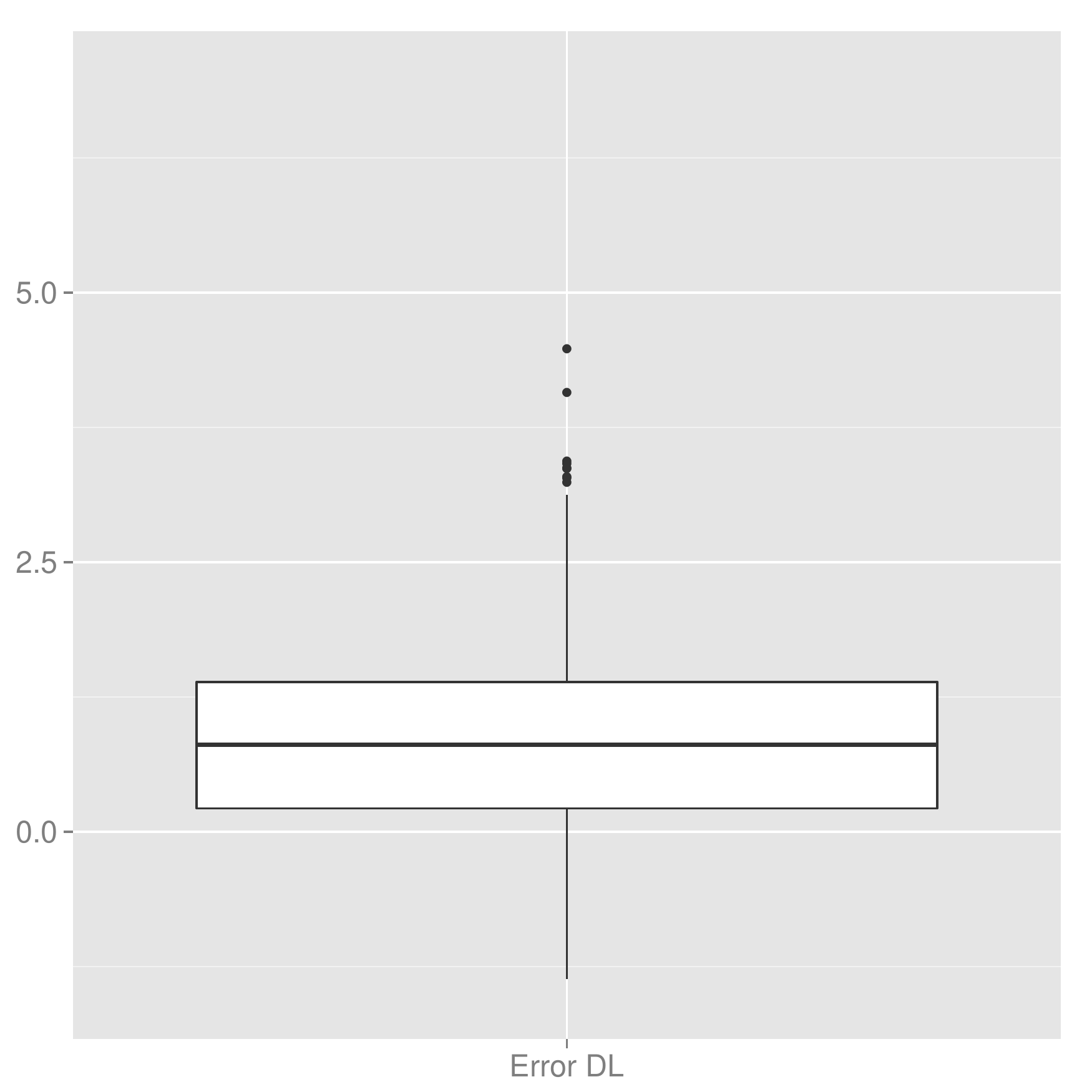}}

\caption{discrete Laplace method. Boxplots comparing the distributions of $\log_{10}\LRf_{\text{DL}}$ and $\widehat{\log_{10}\LRf_{\text{DL}}}$ (left) and the error $e_{\text{DL}}=\widehat{\log_{10}\LRf_{\text{DL}}}$- $\log_{10}\LRf_{\text{DL}}$ (2nd column).}\label{figbru}
\end{figure}

The error of the discrete Laplace method can be defined as $e_{\text{DL}}:=\widehat{\log_{10}{\LRf_{\text{DL}}}}-\log_{10}{\LRf_{\text{DL}}}$. It measures how much the estimated distribution differs from the true one. Table~\ref{tab1f} and Figure~\ref{figbru} (right part) show the distribution of the error. One can see that it can attain up to about 4 orders of magnitude. The distribution of the error is mostly located on positive values, which means that, more often than not, $\widehat{\log_{10}{\LRf_{\text{DL}}}}$ overestimates $\log_{10}{\LRf_{\text{DL}}}$. The standard deviation of the error is small, thereby $e_{\text{DL}}$ does not move too much away from its mean, which is about 0.842.

Motivated by the discussion of Section~\ref{levels}, we now analyze the different levels of uncertainty which affect the error. 
The second level of uncertainty is introduced when the discrete Laplace model, along with all its set of assumptions, is chosen to model the distribution of single locus haplotypes, which in reality do not follow a discrete Laplace distribution. 

The third level of uncertainty pertains to the estimation of the parameters of the model ($c$, $p$, $y$, $\tau$). Here, the databases used to estimate the parameters of the discrete Laplace model are probably too small ($N=100$) with regard to 7 loci.

To decrease both sources of error, one can reduce the number of analyzed loci to 3. The population becomes less sparse, and the databases big enough. Indeed, we performed this experiment and the error decreased a great deal. However, the basic level of uncertainty (see Section~\ref{levels}) is increased inasmuch the data becomes less effective to discern between the two hypotheses. On the other hand, the same experiment with 10 loci leads to obtain more powerful likelihood ratios, but less precise.

The second level of uncertainty can be made harmless assuming an infinite number of subpopulations, since in this way the model will perfectly fit any population. However, this solution will increase the number of parameters, along with the third level of uncertainty. 

It is worth underlining that the results of our simulations do not mean that the discrete Laplace method is wrong on the whole, but they show that a blind use of this method is dangerous. We are applying this method to the specific case of the rare haplotype match, using a database of size 100, and a rather sparse population: maybe this method was never intended to be used for such small databases, and maybe it can be modified in more clever ways to that purpose.

\section{The Generalized-Good method}\label{gg}
Based on \citet{good:1953}, we now propose a nonparametric estimator for the weight of evidence. This is a very good example of data reduction, since $\mathcal{D}$ is here reduced to a greater extent than it was done for the discrete Laplace method.
Indeed, the specific haplotype $x$ of the crime stain and of the suspect is ignored, retaining only the fact that they match and the fact that this profile has not been observed yet in the database.

Stated otherwise, 
\begin{itemize}
\item $D_{\text{GG}}$ = the haplotype of the suspect matches the haplotype of the crime stain and it is not in the database.
\end{itemize}

Consider the following mathematical description: the database of size $N$ can be seen as an i.i.d.\ sample $(Y_1,Y_2, ..., Y_N)$ from species $\{1,2, ...,S\}$, with probabilities $(p_1, p_2,... p_S)$.
Hence, the suspect's profile can be though of as the $N+1$st i.i.d.\ observation.
The crime stain's profile is the $N+2$nd observation. According to the defense it is again an i.i.d.\ draw from $(p_1, p_2,... p_S)$, while according to prosecution it is equal to the value of $Y_{N+1}$, with probability one.

The likelihood ratio for this reduction of the data can be thus written as
 $$\LRf_{\text{GG}}=\frac{\Prf( Y_{N+1}\notin  \mathcal{Y}_N, Y_{N+1}=Y_{N+2} \mid h_p) }{\Prf(Y_{N+1}\notin  \mathcal{Y}_N, Y_{N+1}=Y_{N+2}\mid h_d)} =\frac{\Prf( Y_{N+1}\notin \mathcal{Y}_N \mid \textcolor{black}{h_p}) }{\Prf(Y_{N+1}\notin  \mathcal{Y}_N, Y_{N+1}=Y_{N+2} \mid \textcolor{black}{h_d})}. $$

From now on, we are presenting results regarding a general database size $N>2$, and general random variables $Y_1, ..., Y_N$, i.i.d.\ from $(p_1,p_2,...,p_S)$. The following notation is used:\begin{align*}
\theta_1(N;p_1,p_2,...,p_S) := &\, \Prf(Y_N\notin \{Y_1,Y_2,..., Y_{N-1}\}),\\
\theta_2(N;p_1,p_2,...,p_S) := &\, \Prf(Y_N\notin \{Y_1,Y_2,..., Y_{N-2}\} , Y_{N}=Y_{N-1}).
\end{align*}
To make the notation less cumbersome we will use
\begin{align*}
&\mathcal{Y}_{N}:=(Y_1, Y_2, ..., Y_{N}),\\
&\mathcal{Y}_{i, N}:=(Y_1, Y_2, ..., Y_{i-1}, Y_{i+1},..., Y_{N}),\\
&\mathcal{Y}_{(i,j), N}:=(Y_1, Y_2,..., Y_{i-1}, Y_{i+1},..., Y_{j-1}, Y_{j+1},..., Y_{N}), \quad \forall i<j.
\end{align*}
Moreover, for any random variable $Y$, and any couple of sets $A$ and $B$, $\mathbf{1}_{A\cap B^\mathsf{c}}(Y)$ is a random variable which has value 1 if $Y$ belongs to the set $A$ and not to the set $B$, and zero otherwise.

\begin{teo}
An unbiased estimator for $\theta_1(N;p_1,p_2,...,p_S)$ is $\widehat{\theta}_1(N)=N_1/N$, where $N_1$ is the number of singletons in the database.
\end{teo}

\begin{proof}
\begin{align*}
\theta_1(N;p_1,p_2,...,p_S)&= \Prf(Y_N\notin \mathcal{Y}_{N-1})
=\mathbb{E}(\mathbf{1}_{(\mathcal{Y}_{N-1}) ^\mathsf{c}}(Y_N))=\frac{1}{N}\sum_{i=1}^{N}\mathbb{E}(\mathbf{1}_{(\mathcal{Y}_{i, N })^\mathsf{c}}(Y_i))\\
&=\mathbb{E}\left(\frac{1}{N}\sum_{i=1}^{N}\mathbf{1}_{(\mathcal{Y}_{i, N })^\mathsf{c}}(Y_i)\right)=\mathbb{E}\left(\frac{N_1}{N}\right),
\end{align*}
where the last equality is due to the fact that the function $\mathbf{1}_{(\mathcal{Y}_{i, N })^\mathsf{c}}(Y_i)$ has value 1 for every singleton of the database: the sum is thus the number of singletons ($N_1$). \qedhere

\end{proof}

\begin{teo}
An unbiased estimator for $\theta_2(N;p_1,p_2,...,p_S)$ is $\widehat{\theta}_2(N)=2 N_2/N(N-1)$, where $N_2$ is the number of doubletons in the database.
\end{teo}

\begin{proof}
\begin{align*}
\theta_2(N;p_1,p_2,...,p_S)&= \Prf(Y_N\notin \{\mathcal{Y}_{N-2 } \}, Y_N=Y_{N-1})=\mathbb{E}(\mathbf{1}_{\{Y_{N-1}\cap (\mathcal{Y}_{N-2 })^\mathsf{c}\}}(Y_N))\\
&=\frac{2}{N(N-1)}\sum_{i < j}\mathbb{E}(\mathbf{1}_{\{ Y_j\cap (\mathcal{Y}_{(i,j),N})^\mathsf{c}\}}(Y_i))
=\mathbb{E}\left(\frac{2}{N(N-1)}\sum_{i < j}\mathbf{1}_{\{ Y_j\cap (\mathcal{Y}_{(i,j),N})^\mathsf{c}\}}(Y_i)\right)\\
&=\mathbb{E}\left(\frac{2N_2}{N(N-1)}\right), 
\end{align*}
where the last equality is due to the fact that the function $\mathbf{1}_{\{ Y_j\cap (\mathcal{Y}_{(i,j),N})\mathsf{c}\}}(Y_i)$ has value 1 for each of the $N_2$ doubletons of the database.\qedhere

\end{proof}

The two previous theorems can be easily generalized to $\theta_m$ defined as
$\theta_m(N;p_1, p_2,..., p_S):= \Prf(Y_N\notin \mathcal{Y}_{N-m}, Y_{N}=Y_{N-1}=..=Y_{N-m+1})$.
 
Now we can estimate $\log_{10}\LRf_{\text{GG}}$ in the following way:
\begin{align*}
\log_{10}\LRf_{\text{GG}}&= \log_{10} \frac{\Prf( Y_{N+1}\notin \mathcal{Y}_N \mid h_p) }{\Prf(Y_{N+1}\notin  \mathcal{Y}_N, Y_{N+1}=Y_{N+2}\mid h_d)}  \approx  \log_{10} \frac{\Prf( Y_{N}\notin  \mathcal{Y}_{N-1}) }{\Prf(Y_{N}\notin  \mathcal{Y}_{N-2}, Y_{N}=Y_{N-1})}\\
 &\approx  \log_{10} \frac{\theta_1(N;p_1,p_2,...,p_S)}{\theta_2(N;p_1,p_2,...,p_S)}.
\end{align*}

Thus, we propose the following estimator for the weight of evidence:
\begin{equation}\label{sas3}
\widehat{\log_{10}\LRf_{\text{GG}} }=\log_{10}\frac{\widehat{\theta_1}(N)}{\widehat{\theta_2}(N)}=\log_{10}\frac{(N-1)N_1}{2N_2} \approx \log_{10}\frac{ NN_1}{2N_2}.
\end{equation}

Notice that there are two kinds of approximation steps: a mathematical approximation of $\theta_1(N+1;p_1,p_2,...,p_S)$ with $\theta_1(N;p_1,p_2,...,p_S)$, which should hardly make any difference, for reasonably large $N$, and a statistical estimation of $\theta_1(N;p_1,p_2,...,p_S)$ using an unbiased estimator (and similarly for $\theta_2$).

It is important to underline that, due to Jensen's inequality, the estimators $\log_{10}\widehat{\theta_1}$ and  $\log_{10}\widehat{\theta_2}$ are not unbiased for $\log_{10}\theta_1$ and $\log_{10}\theta_2$, but it will be shown by simulations that $\widehat{\log_{10}\LRf_{\text{GG}}}$ is approximately unbiased for $\log_{10}\LRf_{\text{GG}}$. However, the point is not to find an unbiased estimator, but one with a small error rate.

Notice that in order to estimate $\log_{10}\LRf_{\text{GG}}$ it is not necessary to use all the information contained in the database, but only $N$, $N_1$, and $N_2$, that is the number of singletons and doubletons in the database. The nuisance parameter of the model is the vector  $\theta$ containing the frequencies of the Y-STR haplotypes in the population of interest. $\theta_1$ and $\theta_2$ are functions of $\theta$. 

The limitation of this method is that it cannot be used if $N_2=0$ (this corresponds to an infinite likelihood ratio) and it does not perform well  also in case the number of singletons is very small or zero. We believe it can be improved and extended by smoothing techniques \citep{good:1953, anevski:2013}, but we are going to ignore this problem.

The `$\kappa$-method' of Brenner \citep{brenner:2010} is based on an analogous line of reasoning. It estimates the likelihood ratio as 
$\widehat{\LRf}_{\kappa}=\frac{N^2}{N-N_1}.$
However, in the derivation of this estimator, there is an approximation involved, based on assumptions which are not always satisfied, leading sometimes to anti-conservatism (see also the discussion in \citet{buckleton:2011}, and the answer in \citet{brenner:2014}). In particular, \citet{brenner:2014} provides a pathological population where the approximation does not hold, while showing empirical evidence that for Fisher-Wright populations the condition is fulfilled. 
Our method is, on the other hand, based on a principled derivation of the estimator of equation \eqref{sas3}, which is similar to Brenner's one under the following conditions:
 there are almost only singletons and doubletons in the database, and
 $N_1\gg N_2$.
These assumptions are typically satisfied, explaining why Brenner's method often works. They also constitute a good description of when it does not work.

Lastly, we remark that this method can be generalized in the obvious way, to the case in which the haplotype is indeed in the database.
Moreover, this method is suitable to be directly applied to different kinds of evidence.

\subsection{Quantifying the uncertainty of the \text{GG} method}

As we did in Section~\ref{q-dis}, we want to quantify the uncertainty of this method. One way is to compare the distribution of $$\widehat{\log_{10}{\LRf_{\text{GG}}}}=\log_{10}\frac{NN_1}{2N_2},$$ with the distribution of the ``true'' 
$$\log_{10}{\LRf_{\text{GG}}}=\log_{10} \frac{\Prf( Y_{N+1}\notin  \mathcal{Y}_N) }{\Prf(Y_{N+1}\notin \mathcal{Y}_N \cap Y_{N+1}=Y_{N+2})} :=\log_{10}\frac{\theta_1}{\theta_2}.
$$

Actually, the latter is not a distribution, but a single value, unknown. Again, we pretend that the database of \citet{purps:2014} contains the profiles of the whole population, to find out the `true' $\theta_1$ and $\theta_2$, restricting our simulations to 7 loci. 
To do so, we sample $M$ small databases of size $N=100$, along with two other haplotypes. $\theta_1$ is the proportion of times in which the $(N+1)$st haplotype is a new one (i.e., not one of the previous $N$), and $\theta_2$ is the proportion of times in which the $(N+2)$nd is equal to the $(N+1)$st, and different from the first $N$ observations.
From our simulations, we used $M=100,000$, and we obtained $\theta_1$, $\theta_2$, and $\log_{10}\LRf$ as in Table~\ref{tab1bbig}. 
\begin{table}[htbp]
\begin{center}
    \begin{tabular}{|c|c|c|}
    \hline
 $\theta_1$    & $\theta_2$ & True $\log_{10}\LRf_{\text{GG}}$    \\
  \hline
0.748   	&0.0012	&   2.78				  \\

    \hline 
    \end{tabular}
    \caption{Values of $\theta_1$ and $\theta_2$ and of $\log_{10}\LRf_{\text{GG}}$   obtained by simulations, assuming that the database of \citet{purps:2014} contains the whole population of interest.  }\label{tab1bbig}
\end{center}
\end{table}

The distribution of $\widehat{\log_{10}\LRf_{\text{GG}} }=\log_{10}\frac{NN_1}{2N_2}$ can be obtained by sampling $M=100,000$ databases of size $N=100$. Out of 100,000 databases, 121 had $N_2=0$. They have been removed from the data, and we acknowledge that this choice creates unfairness to the discrete Laplace method. On the other hand, we believe that this occurs frequently enough not to affect very strongly the comparison.

\begin{figure}[htbp]
\medskip

\begin{center}
\includegraphics[width=0.4\textwidth]{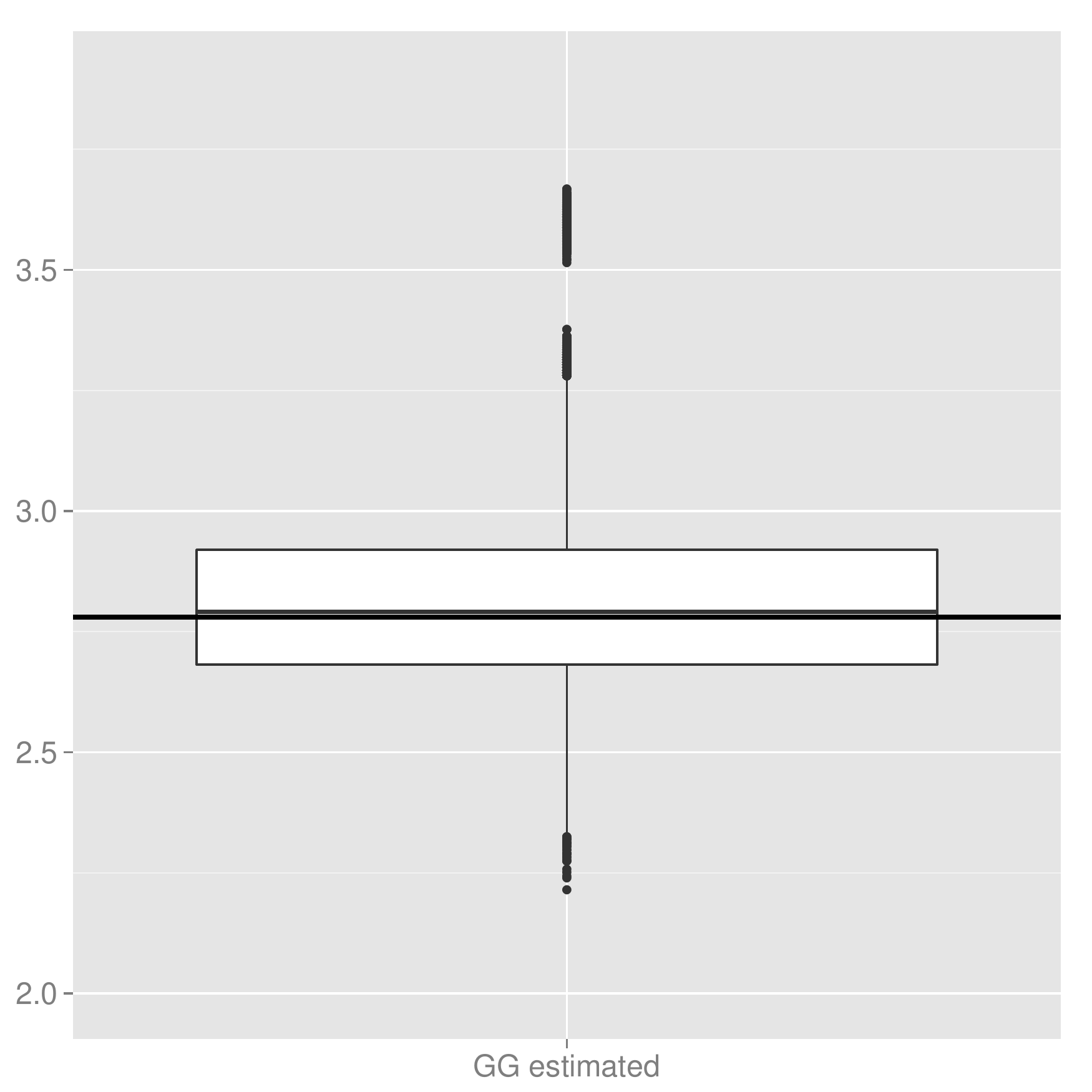}
\end{center}
\caption{Boxplots of the distribution of $\widehat{\log_{10}\LRf_{\text{GG}} }$ around the true value $\log_{10}\LRf_{\text{GG}}$ (black line).}\label{figure3}
\end{figure}

Figure~\ref{figure3} shows the distribution of the estimator  $\widehat{\log_{10}\LRf_{\text{GG}} }$ around the true value (black line).
The error of the Generalized-Good method, defined as $e_{\text{GG}}= \widehat{\log_{10}\LRf_{\text{GG}} }-\log_{10}\LRf_{\text{GG}} $, tells us how much the estimator differs from the true value. 

\begin{table}[htbp]
\begin{center}
    \begin{tabular}{|l|c|c|c|c|c|c|c|}
    \hline
 &  Min   & 1st Qu. & Median & Mean  & 3rd Qu. & Max   & sd   \\
  \hline
  $\log_{10}\LRf_{\text{GG}}$   			&  2.78	&  2.78 	&  2.78 	&  2.78 	&  2.78 	&  2.78	&	0   \\
					$\widehat{\log_{10}\LRf_{\text{GG}}}$    	& 2.215	&   2.682 		&  2.792		&	   2.818	&   2.920	&  3.668 	& 0.198  \\
					Error $e_{\text{GG}}$					& -0.566	& -0.098&  0.0112 &  0.038&  0.14  & 0.887 & 0.198 \\				
   
    \hline 
    \end{tabular}
 
    \caption{Summaries of the distribution of $\widehat{\log_{10}\LRf_{\text{GG}}}$, of $\log_{10}\LRf_{\text{GG}}$, and of the error $e_{\text{GG}}$. }\label{tab1bigd}
\end{center}
\end{table}

Table~\ref{tab1bigd} provides the summaries for $\widehat{\log_{10}\LRf_{\text{GG}} }$, and for the error $e_{\text{DL}}$. 
We don't provide the plots for the distribution of $e_{\text{GG}}$ since they are identical to those in Figure~\ref{figure3}, shifted of $\log_{10}\LRf_{\text{GG}}$.

One can see that the error can attain up to about 0.9 orders of magnitude. The distribution of the error is mostly located on the positive values, which means that, more often than not, $\widehat{\log_{10}{\LRf_{\text{GG}}}}$ overestimates $\log_{10}{\LRf_{\text{GG}}}$. The standard deviation of the error is small, thereby $e_{\text{GG}}$ does not move too much away from the mean, which is about 0.038. 
If compared to the error of the discrete Laplace  method, one can conclude that here we get a better estimator in terms of accuracy, since the error ranges over more restrained values and the standard deviation is much smaller. However, it is important to keep in mind that they are not different estimators of the same quantity, but different estimators of different quantities, since the reduction of data used by the Generalized-Good method, which allows obtaining accuracy in the estimates is less strong to discern between the two hypotheses.
\section{Choosing and comparing methods}\label{trade}
In comparing the two methods one can consider the precision with respect to what the method is trying to estimate, quantified by the errors $e_{\text{DL}}$, and $e_{\text{GG}}$. These errors are due to the two second and third level of uncertainty described in Section~\ref{levels}, and decrease sensibly if data is reduced. This is why, under this aspect, the Generalized-Good is to be preferred to the discrete Laplace, and for the latter a fewer number of loci is to be preferred.
However, it is not correct to believe that the greater the reduction, the better is the method. To reduce means to lose information, and thus to diminish the capability of the method to distinguish between the hypotheses at stake (the first, or basic level of uncertainty).
In order to investigate this loss, one can compare each method to the likelihood ratio $1/f$ (where $f$ is the population frequency of the matching haplotype), which can be considered the hallmark in a population with no substructure. 
Comparing Table~\ref{tab1f} with Table~\ref{tab1bbig} one can see for instance that choosing the Generalized-Good one loses on average around $0.5$ (in logarithmic scale) in terms of strength of data to discriminate between hypotheses. This is a small disadvantage for the prosecution, while everybody gain in terms of precision with respect to the true $\log_{10}{\LRf_{\text{GG}}}$.
As a last remark, we invite the reader to realize that the discrete Laplace method is better inasmuch it can always be used. On the other hand, for the Generalized-Good, we had to remove 121 experiments where $N_2=0$.

 \section{Remark and conclusion}

The aim of this paper could, at first sight, be considered that of offering two additional frequentist methods to address the issue of the likelihood ratio calculation in the case of a rare Y-STR haplotype match.
However, a careful reader may have realized that these methods also constitute two interesting opportunities to show and apply the guidelines exposed in the opening sections. In particular, two important facts are pointed out in Sections~\ref{LR} and \ref{levels}: first, it is more sensible to talk about ``a'' likelihood ratio instead of ``the'' likelihood ratio, and second, a quantification of the error involved in the estimation is to be provided along with the estimate of the likelihood ratio.

Moreover, it is explained that sometimes it is possible to the break down the data to be evaluated into $E$ (which is sufficient for $H$) and $B$ (which is irrelevant for $H$). 
The discrete Laplace method (developed in Section~\ref{DLm}) is a good example where this distinction can be done, while the same is not true for the Generalized-Good method (Section~\ref{gg}). 

Lastly, this paper wants to get across the message that reducing the data to a smaller extent is sometimes not only necessary, but also desirable in terms of exactitude of the estimates, as proved by the comparison between the discrete Laplace method (less reduction, less precision of the estimates) and the Generalized-Good method (stronger reduction, more precision of the estimates). In this respect we disagree with \citet{buckleton:2011} who, talking about Brenner's method, state that `there is a merit focussing in the type or name of a lineage marker''. Although we agree that ``such ignorance of type implies a substantial loss of information'', it may allow a large gain in precision.

The take home message is that choosing the best method is clearly a very delicate task. One has to consider many different aspects, and look for a compromise which is acceptable for the specific application at hand. It is important to realise that in this paper we study a very extreme situation with very small databases and a possibly unrealistic population, for which the Generalized-Good seemed to be the best compromise. Clearly, there are no possible general conclusions to be given, other than at each new situation one has to reconsider all these aspects, and weigh them.

\section*{Acknowledgements}

The Generalized-Good method described here was suggested by Richard Gill and presented in several conference lectures, see for instance \url{http://www.slideshare.net/gill1109/the-fundamental-problem-of-forensic-statistics-38322519}.
I am indebted to Charles Brenner (the first to use the `fundamental problem' name), and to Ronald Meester, for the useful discussions about this paper, which lead to many improvements.
This research was supported by the Swiss National Science Foundation, through grants no.\ 105311-1445570 and 10531A-156146/1, and carried out in the context of a joint research project, supervised by Franco Taroni (University of Lausanne, Ecole des sciences criminelles,  Facult\'e de droit, des sciences criminelles et d'administration publique), and Richard Gill (Mathematical Institute, Leiden University).


\end{document}